\documentclass{article}

\usepackage{graphicx} 
\usepackage{amsthm}
\usepackage{commath, mathtools, amsfonts, amsmath, amssymb}
\usepackage{graphicx}
\usepackage{float}
\usepackage{caption, subcaption}
\usepackage{tikz}
\usepackage{thmtools}
\usepackage{algorithm}
\usepackage{algpseudocode}
\usepackage{mathdots, nicematrix}

\usepackage{hyperref}
\hypersetup{
    colorlinks=true,
    linkcolor=blue,
    filecolor=blue,      
    urlcolor=blue,
    citecolor=blue, 
    pdftitle={Overleaf Example},
    pdfpagemode=FullScreen,
    }

\newtheorem{theorem}{Theorem}[section]
\newtheorem{proposition}{Proposition}[section]
\newtheorem{corollary}{Corollary}[section]
\newtheorem{lemma}{Lemma}[section]
\newtheorem*{theorem*}{Theorem}

\theoremstyle{definition}
\newtheorem{definition}{Definition}[section]

\newcommand*{\C}{\mathbb{C}}
\newcommand*{\Q}{\mathbb{Q}}
\newcommand*{\R}{\mathbb{R}}

\newcommand*{\Z}{\mathbb{Z}}
\newcommand*{\F}{\mathbb{F}}
\renewcommand*{\O}{\mathfrak{O}}

\renewcommand*{\Re}{\operatorname{Re}}

\title{A Fast Multiplication Algorithm and RLWE--PLWE Equivalence for the Maximal Real Subfield of the $2^r p^s$-th Cyclotomic Field}
\author{\small{Wilmar Bolaños, Antti Haavikko, Rodrigo M. Sánchez-Ledesma}}
\date{February 2025}

\begin{document}

\maketitle

\begin{abstract}
    This paper proves the RLWE--PLWE equivalence for the maximal real subfields of the cyclotomic fields with conductor $n = 2^r p^s$, where $p$ is an odd prime, and $r \geq 0$ and $s \geq 1$ are integers. In particular, we show that the canonical embedding as a linear transform has a condition number bounded above by a polynomial in $n$. In addition, we describe a fast multiplication algorithm in the ring of integers of these real subfields. The multiplication algorithm uses the fast Discrete Cosine Transform (DCT) and has computational complexity $\mathcal{O}(n \log n)$. Both the proof of the RLWE--PLWE equivalence and the fast multiplication algorithm are generalizations of the results in \cite{ahola2024fast}, where the same claims are proved for a single prime $p = 3$.
\end{abstract}


\section{Introduction}

Lattice-based cryptography has emerged as a promising candidate for post-quantum cryptography, offering potential resilience against attacks by quantum computers. The \emph{Learning with Errors} (LWE) problem and its variants, \emph{Ring-LWE} (RLWE) and \emph{Polynomial-LWE} (PLWE), are fundamental hardness assumptions in the field of post-quantum cryptography (PQC). These problems involve finding approximate solutions to noisy linear systems, and RLWE and PLWE have a performance advantage over the unstructured LWE due to their inherent algebraic structure. 

The origins of RLWE and PLWE can be traced back to the seminal works of Stehlé et al. \cite{SSTX:2009:EPK} and Lyubashevsky et al. \cite{LPR:2013:ILL}. Since their discoveries, the relationship between the RLWE and PLWE problems has been an active area of research. RLWE is generally considered to be more secure because of its connection to well-established lattice problems. Also, the abstract setting of the RLWE problem is fitting for theoretical proofs. On the other hand, PLWE offers practical advantages regarding implementation and performance, particularly for multiplication operations. 

When it comes to other areas of modern cryptography, fast multiplication algorithms play a substantial role also in \emph{Homomorphic Encryption} (HE), especially in schemes that derive from the structured LWE variants. In general, HE schemes use larger parameters than PQC schemes and hence benefit more from improvements in the asymptotic complexity of multiplication.

Prior research has established the equivalence between RLWE and PLWE for specific families of number fields and polynomials. For non-cyclotomic number fields, the work of Ahola et al. \cite{ahola2024fast} demonstrated this equivalence for the maximal totally real subfield of the $2^r 3^s$-th cyclotomic field, along with a fast multiplication algorithm based on the Discrete Cosine Transform (DCT). Also, the work in \cite[Section 4]{RSW:2018:RPP} focused on two large non-cyclotomic families of polynomials and proved that the related Vandermonde matrices have polynomially bounded condition numbers.

In the cyclotomic case, notable contributions include the work of Ducas and Durmus \cite{DD:2017}, who proved the equivalence for cyclotomic number fields of degree $2^k p$ or $2^kpq$, where $p$ and $q$ are primes and $q \neq p$. Furthermore, Di Scala et al. \cite{SSS:2024:RPNE} showed that RLWE and PLWE are not equivalent in general if no restrictions are imposed on the conductor of the cyclotomic field. However, Blanco-Chacón \cite{Blanco:2020:REC} proved that the RLWE--PLWE equivalence holds if the conductor of the cyclotomic field is divisible by a bounded number of primes. Later, Araujo \cite{dearaujo2025condition} showed the equivalence for any cyclic number field of odd prime degree, which by the Kronecker--Weber theorem is a subfield of some cyclotomic field.

Regarding the hardness of the structured lattices, \cite{cramerducaswesolowsky} uses Stickelberger ideals from class field theory to prove the existence of a quantum polynomial-time attack against ideal lattices of cyclotomic number fields. Their results assume the Generalized Riemann Hypothesis and impose some class-field theoretical restrictions on the cyclotomic fields. The exact claim is the following:
\begin{theorem*}
    There exists a quantum polynomial-time algorithm, that for a cyclotomic number field $K$ of prime power conductor and any ideal $\mathfrak{a}$ of $\O_K$, returns an element $v \in \mathfrak{a}$ with Euclidean norm
    \begin{align*}
        ||v|| \leq \operatorname{N}(\mathfrak{a})^{1/n} \operatorname{exp}(\tilde{\mathcal{O}}(\sqrt{n})).
    \end{align*}
\end{theorem*}
Note that the ideals in the ring of integers of the maximal real subfields of cyclotomic fields are not ideals in the ring of integers of the cyclotomic field containing the subfield. Therefore, the quantum attack of \cite{cramerducaswesolowsky} does not apply to the maximal real subfields.

This paper extends the results of \cite{ahola2024fast} by generalizing the fast multiplication algorithm and the equivalence between RLWE and PLWE to the maximal real subfield of the $2^r p^s$-th cyclotomic field, where $p$ is an arbitrary odd prime. This generalization comes with many technical lemmas and verifications guaranteeing that the security and speed of the RLWE or PLWE schemes are not compromised. The main contributions of this paper are:
\begin{enumerate}
    \item A proof of the equivalence between RLWE and PLWE for maximal real subfields of cyclotomic fields with conductor $n = 2^r p^s$.
    \item A fast multiplication algorithm with quasilinear complexity for elements in the ring of integers of these subfields, utilizing the Discrete Cosine Transform.
    \item Explicit algorithms for the change of basis between the power basis and the modified Chebyshev basis, enabling efficient computations regardless of the basis.
\end{enumerate}

While the use of maximal real cyclotomic polynomials in PLWE is new, they are not merely an artifact of \cite{ahola2024fast} and this work. The study of the relevant security properties of PLWE schemes under these polynomials, along with the development of fast multiplication techniques, is a result of a growing interest in finding applicable non-cyclotomic PLWE instances. Regarding applications, the support for the maximal real polynomials has already been incorporated into new cryptographic libraries (see LATTIGO \cite{lattigo}).

The structure of the paper is as follows. In Section \ref{sec:preliminaries}, we introduce the necessary background on maximal real cyclotomic polynomials, matrix norms, and properties of certain cosine matrices related to the canonical embedding of the number fields at play. We give bounds on the Frobenius norms of these cosine matrices and use the result in the proofs of the two main theorems. Section 2 also provides the full definitions of the RLWE and PLWE distributions.

Section \ref{sec:reduction_formulas} covers the $\Q$-linear dependencies between elements of the form $2\cos(2 \pi j / n)$ for two cases of conductor, namely $n = p^s$ and $n = 2^r p^s$. Throughout the paper, we study these two cases for $n$ separately. We derive explicit formulas for the minimal polynomials $\Psi_n(x)$ of the primitive elements $\psi_n = 2 \cos(2\pi/n)$ given in the basis of modified Chebyshev polynomials $V_j(x)$ that are introduced in Section \ref{sec:preliminaries}. We describe in detail the reduction of polynomials to small degree representatives in the polynomial quotient ring $\Z[x] / (\Psi_n(x))$, again using the Chebyshev basis. These formulas will play an integral part in the proof of all four main theorems of the paper.

Section \ref{sec:plwe_rlwe_equivalence} tackles the PLWE--RLWE equivalence of the maximal real subfields of cyclotomic fields with conductor $n$ as above. The two main theorems of Section \ref{sec:plwe_rlwe_equivalence} combine to the following.
\begin{theorem*}
    Let $p \geq 3$ be a prime, $s \geq 1$, and $r \geq 0$. Fix $n = 2^r p^s$. Then PLWE and RLWE are equivalent for the maximal real subextension of the $n$-th cyclotomic field.
\end{theorem*}
We prove that the condition number of the canonical embedding matrix is bounded above by a polynomial in $n$. By the results in \cite{DD:2017,RSW:2018:RPP}, we know that this is sufficient to show the equivalence of the PLWE and RLWE problems.

On top of the theoretical advancements, Section \ref{sec:fast_multiplication_via_DCT} introduces a method for quasilinear multiplication of two elements in the polynomial quotient ring $\Z[x] / (\Psi_n(x))$. The fast multiplication in $\Z[x]$ makes use of the \emph{Discrete Cosine Transform} to compute a linear convolution in the modified Chebyshev basis. Furthermore, we show that computing the remainder in the quotient ring can be done with linear complexity in $n$ by using the reduction formulas of Section \ref{sec:reduction_formulas}. Lastly, we provide references that show that the change of basis computations between the canonical power basis and the Chebyshev basis can be done also with quasilinear complexity. These findings are stated as the second main theorem.
\begin{theorem*}
For $n = 2^r p^s$ with $r \geq 0$ and $s \geq 1$, given two polynomials $a(x)$, $s(x) \in \Z[x] / (\Psi_n(x))$ in the power basis, their product $a(x) \cdot s(x) \in \Z[x] / (\Psi_n(x))$ can be computed with asymptotic complexity $\mathcal{O}(n \log n)$.
\end{theorem*}
Section \ref{sec:computational_analysis} defines a computational framework for studying how practically robust are the polynomials $\Psi_n(x)$ against known algebraic attacks against PLWE. We are particularly interested in the attacks that exploit the information about the roots of $\Psi_n(x)$. These attacks are constructed from the results of \cite{ELOS:2015:PWI, ELOS:2016:RCN, barbero2023cryptanalysis}, which were later generalized in \cite{BDM:2024:GARB}, laying the ground work to avoid special conditions over the distribution of the errors, along with applicability for higher degree extensions of finite fields. In Section \ref{sec:computational_analysis}, computational calculations are performed on a number of targeted samples, which are constructed to mimic \textit{cryptographically relevant} schemes, to obtain a ratio of how many of these polynomials are affected by the attacks. These numbers are also computed for the matching cyclotomic instances to compare the robustness of the two families. The results show that the maximal real polynomials do not suffer from any meaningful increased vulnerability against this algebraic approach.

\section{Preliminaries}
\label{sec:preliminaries}
This section will introduce the necessary background to follow the main theorems and lemmas of this paper. Our goal is to familiarize the reader with the construction of the maximal real subfields of cyclotomic fields and its ring of integers. Further, we introduce an alternative polynomial basis that will be the natural choice in the computations and derivations of the minimal polynomials of these fields. Lastly, we provide a series of lemmas as a preparation for the proof of the main Theorems \ref{thm:PLWE_RLWE_equivalence_p_to_s} and \ref{thm:PLWE_RLWE_equivalence_p_to_s_pow_2}.

\subsection{Maximal real subfields of cyclotomic fields}
In this work, we study the PLWE problem and general multiplication in the ring of integers of a \emph{maximal real subfield of a cyclotomic field}. 
\begin{definition}
    Let $n \geq 3$ and $\zeta_n = e^{2 \pi i / n}$ be a primitive $n$-th root of unity in $\C$. The \emph{maximal real subfield of a cyclotomic field} $\Q(\zeta_n)$ is defined as the intersection of the cyclotomic fields with the reals,
    \begin{align*}
        \Q(\zeta_n)^+ := \Q(\zeta_n) \cap \R.
    \end{align*}
\end{definition}

These number fields are generated by the real number $\psi_n = 2 \cos(2\pi / n)$ and the degree of the extension over $\Q$ is $\phi(n) / 2$. In addition, these fields are the fixed fields of the complex fields $\Q(\zeta_n)$ under the complex conjugation map. We will denote the minimal polynomial of $\psi_n$ over $\Q$ by $\Psi_n(x)$. In the literature, the polynomials $\Psi_n(x)$ are known as \emph{maximal real cyclotomic polynomials} or simply \emph{real cyclotomic polynomials}.
\begin{lemma}
The field $\Q(\zeta_n)^+$ is generated by $\psi_n = \zeta_n + \zeta_n^{-1} = 2 \cos(2 \pi / n)$, that is, $\Q(\zeta_n)^+ = \Q(\psi_n)$. Moreover, $[\Q(\zeta_n) : \Q(\zeta_n)^+] = 2$ and 
$$[\Q(\zeta_n)^+ : \Q] = \frac{\phi(n)}{2}.$$
\end{lemma}
\begin{proof}
    The properties follow from $\Q(\zeta)^+$ being by definition the maximal real subextension and the fact that 
    $$x^2 - (\zeta_n + \zeta_n^{-1})x + 1 = (x -\zeta_n)(x - \zeta_n^{-1})$$ 
    is the minimal polynomial of $\zeta_n$ over $\Q(\zeta_n)^+$.
\end{proof}

Finally, for the PLWE setting we will need the following theorem, since the evaluation isomorphism $\Z[x] / (\Psi_n(x)) \to \Z[\psi_n]$ will help us construct the reduction formulas and dependencies in Section \ref{sec:reduction_formulas}.
\begin{theorem}
    The ring of integers of $\Q(\psi_n)$ is $\Z[\psi_n]$.
\end{theorem}
\begin{proof}
    See \cite[Proposition 2.16]{washington2012introduction}.
\end{proof}

\subsection{Modified Chebyshev polynomials}
Chebyshev polynomials and the cosine function are connected. For the degree $n$ Chebyshev polynomial $T_n(x)$, the cosine of a multiple of an angle can be written as polynomial in $\cos(\theta)$ as $T_n(\cos(\theta)) = \cos(n \theta)$. Since our element of interest $\psi_n = 2\cos(2\pi / n)$ has an extra factor of 2, we introduce a modified family of Chebyshev polynomials that exhibit similar behavior with $2\cos(2\pi / n)$.
\begin{definition}
    Let $T_n(x)$ be the Chebyshev polynomial of the first kind of degree $n$. We define the family of polynomials $V_n(x)$ with $V_0(x)= 1$ and 
    \begin{align*}
        V_i(x) = 2T_i(x/2) \quad \text{for } i \geq 1
    \end{align*}
     to be the \emph{modified Chebyshev polynomial of degree} $i$.
    \label{def:modified_chebyshev_polynomials_V}
\end{definition}
From the definition, the family of modified Chebyshev polynomials satisfy the following property,
\begin{align*}
    V_n(2\cos(\theta)) = 2\cos(n\theta) \quad \text{for any } \theta  \text{ and } n \geq 1.
\end{align*}
For the following properties of the $V_n(x)$ and more, we refer the reader to \cite{loper2016resultants}. 
\begin{proposition}
\label{prop:properties_of_V}
For any $m,n \geq 1$ we have,
\begin{align*}
    V_m(2) = 2  \quad \text{and} \quad  V_n(V_m(x)) = V_{mn}(x).
\end{align*}
\end{proposition}

\begin{proposition}
\label{prop:Vn_times_Vm}
For any $m,n \geq 1$ such that $n \neq m$, we have the relation 
$$V_n(x)V_m(x) = V_{m+n}(x) + V_{|m-n|}(x),$$
and if $n = m$, we have 
$$V_n(x)^2 = V_{2n}(x) + 2.$$
\end{proposition}
Alternatively, the polynomials $V_n(x)$ can be defined recursively by
$$V_n(x) = x V_{n-1}(x) - V_{n-2}(x) \mbox{ \; for \;} n \geq 3,$$
with the initializing sequence
$$V_0(x)=1 ,\; V_1(x) = x , \; V_2(x) = x^2 - 2.$$
Note that this recursive definition coincides with the Lucas sequence of the second kind $V_n(P,Q)$ when $P = x$ and $Q = 1$.

In addition, it is easy to see that for any $n$ and $m = \phi(n) / 2$, the set
\begin{align}
\label{eq:V_basis}
    V := \{V_0(x), V_1(x), \hdots, V_{m-1}(x)\}
\end{align}
is a basis for $\mathfrak{O} = \Z[x] / (\Psi_n(x))$ as all the polynomials $V_i(x)$ are monic with integer coefficients and degree $i$.

\subsection{Matrix norms}
Our main tool in proving the equivalence of the RLWE and PLWE problems will be the \emph{condition number} of the Vandermonde matrix that relates the two problems. We provide the necessary definitions and relations below.
\begin{definition}
Let $A\in\mathrm{GL}_n(\mathbb{C})$ be an invertible square matrix with complex entries and $|| \cdot ||$ a matrix norm. The \emph{condition number} of $A$ with respect to the matrix norm $|| \cdot ||$ is defined as
$$
\kappa_{|| \cdot ||}(A)=||A|| \cdot ||A^{-1}||.
$$
\end{definition}
Note that the condition number is sub-multiplicative whenever the corresponding matrix norm is sub-multiplicative.

We will be working with condition numbers defined with respect to the sub-multiplicative \emph{Frobenius norm},
$$|| A ||_F := \sqrt{\sum_{i,j} |a_{i,j}|^2},$$
and the induced 2-norm known as the \emph{spectral norm},
$$|| A||_2 := \sup_{||x||_2 \neq 0} \frac{||Ax||_2}{||x||_2}.$$

In these cases, we denote the condition number of $A$ in the Frobenius norm and the spectral norm by $\kappa_{F}(A)$ and $\kappa_2(A)$, respectively.
It is well known that the norms satisfy the relation
\begin{equation}
    ||A||_2 = \sigma_{\max}(A) \; \leq  \; ||A||_F = \sqrt{\sum_{i} \sigma_i(A)^2} \;  \leq \; \sqrt{n} ||A||_2,
\end{equation}
where $\sigma_i(A)$ are the singular values of $A$ and $\sigma_{\max}(A)$ represent the maximal of these. From the definition of the condition numbers it follows that 
\begin{equation}
\label{eq:cond_num_ineq}
    \kappa_2(A) \leq \kappa_F(A) \leq n \cdot \kappa_2(A).
\end{equation}

\subsection{Cosine matrices}
Next, we will present two cosine matrices that will play an essential role in the proof of the PLWE--RLWE equivalence in Section \ref{sec:plwe_rlwe_equivalence}. The two cases that we cover are a prime power conductor $n = p^s$ and its product with a power of two $n = 2^r p^s$. We will show that the condition numbers of the cosine matrices are polynomially bounded in $n$.
\subsubsection{Case \texorpdfstring{$n = p^s$}{n = p\^s}}
Fix an odd prime number $p$, and consider the prime power $n=p^s$ with $s\geq 1$. Let $N = (p^s - 1) / 2$ and define a grid of points $x_i = 2\cos (2\pi i / n)$ on the real line. With this setting, we introduce a \emph{cosine matrix} for $n=p^s$.
\begin{definition}
\label{def:cosine_matrix_p_to_s}
The cosine matrix $\mathbf{C}_{N+1}$ is a square matrix indexed by $i$ and $j$, where $i=1,\dots, N+1$ , $j=0,1, \dots , N$. We set  $\left(\mathbf{C}_{N+1}\right)_{i,j} = c_{ij} = V_j(x_i)$ for $i \leq N$, and for the last row $i = N+1$, we set $c_{ij} = 1$. This defines the matrix
\begin{align*}
    {\mathbf{C}_{N+1}} &=
    \begin{bmatrix}
        1 & 2\cos(2\pi/ p^s) & 2 \cos(2 \pi 2 / p^s) & \hdots & 2\cos(2\pi N /p^s) \\
        \vdots & \vdots & \vdots & \hdots  & \vdots \\
        1 & 2\cos(2\pi i/ p^s) & 2 \cos(2 \pi 2 i / p^s) & \hdots & 2\cos(2\pi N i /p^s) \\
        \vdots & \vdots & \vdots & \ddots & \vdots \\
        1  & 1 & \hdots & \hdots &  1
    \end{bmatrix}.
\end{align*}
\end{definition}
To prove that the condition number of the matrix $\mathbf{C}_{N+1}$ is polynomially bounded in $N = (p^s - 1) / 2$, we will need the following two lemmas.
\begin{lemma}
\label{lemma:orthogonal}
    For any $1 \leq \sigma \leq p^s -1$, we have the following sum
    \begin{equation} 
    1 + \sum_{j=1}^{N} 2\cos \left( \frac{2\pi \sigma}{p^s} j\right) = 0.
    \end{equation}
\end{lemma}

\begin{proof}
Let $\zeta = e^{2\pi \sigma / p^s}$ be an $n$-th root of unity in $\C$. Then applying the identity 
$$\cos\left( \frac{2\pi \sigma j}{n} \right) = \cos 
 \left(\frac{2\pi \sigma (n - j)}{n} \right)$$
to the sum gives us
\begin{align*}
    1 + \sum_{j=1}^{N} 2\cos \left( \frac{2\pi \sigma j}{p^s} \right) &= 1 + \sum\limits_{j=1}^{2N} \cos \left( \frac{2\pi \sigma j}{p^s} \right) \\
    & =  \Re\left( 1 + \zeta + \zeta^2 + \dots+ \zeta^{2N} \right) \\
    & =  \Re \left(\frac{\zeta^{p^s}-1}{\zeta - 1} \right) \\
    & = 0.
\end{align*}

\end{proof}

\begin{lemma}
\label{lemma:symmetric_dct}
Let $p$ be a prime number, $N =(p^s -1) / 2$, and $\mathbf{C}_{N+1}$ the cosine matrix of dimension $N+1$. Then the Gram matrix of the rows of $\mathbf{C}_{N+1}$ is
\begin{equation}
\mathbf{C}_{N+1} \mathbf{C}_{N+1}^T = \begin{pmatrix}
    2N & -1 & \dots & -1 & 0\\
    -1 & 2N & \dots & -1 & 0 \\
    \vdots & \vdots & \ddots & \dots& \vdots \\
    -1 & -1 & \dots & 2N & 0 \\
    0 & 0 & \dots& 0 & N+1
\end{pmatrix}.
\end{equation}
\end{lemma}

\begin{proof}
First, let us consider row and column indices $u,v \leq N$. For these, we simplify and use Lemma \ref{lemma:orthogonal} to get
\begin{align*}
    (\mathbf{C}_{N+1} \mathbf{C}_{N+1}^T)_{uv} &=  1 + 4\sum_{j=1}^{N} \cos\left(\frac{2\pi u}{p^s} j\right) \cos\left(\frac{2\pi v}{p^s} j\right)   \\
    &=1 + 2\sum_{j=1}^{N} \cos\left(\frac{ 2\pi j (u+v)}{p^s}\right ) + \cos\left(\frac{ 2\pi j(u-v)}{p^s}\right ) \\
    &= 1 + \sum_{j=1}^{N} 2\cos\left(\frac{ 2\pi j (u+v)}{p^s}\right ) + \sum_{j=1}^{N} 2\cos\left(\frac{ 2\pi j (u-v)}{p^s}\right )\\
    &= \sum_{j=1}^{N} 2\cos\left(\frac{ 2\pi j (u-v)}{p^s}\right ).
\end{align*}

Now, if $u \neq v$ by the sum simplifies to 
$(\mathbf{C}_{N+1} \mathbf{C}_{N+1}^T)_{uv} = -1$. Otherwise, $(\mathbf{C}_{k+1} \mathbf{C}_{N+1}^T)_{uv} = 2N = p^s - 1.$ Finally for the last column, if  $u=N+1$ and $v \leq N$, then again by Lemma \ref{lemma:orthogonal} $(\mathbf{C}_{N+1} \mathbf{C}_{N+1}^T)_{uv} = 0$. If $u = v = N+1$, we have $(\mathbf{C}_{N+1} \mathbf{C}_{N+1}^T)_{uv} = N + 1$.
\end{proof}

Now, we can prove the polynomial bound for the condition number.
\begin{lemma}
\label{lemma:cond_number_of_C_N+1}
    The condition number of the cosine matrix  $\mathbf{C}_{N+1}$ satisfies 
    \begin{align*}
        \kappa_F(\mathbf{C}_{N+1}) < \sqrt{2}(N+1).
    \end{align*}
\end{lemma}

\begin{proof}
From Lemma \ref{lemma:symmetric_dct}, we see that the eigenvalues of $\mathbf{C}_{N+1}\mathbf{C}^T_{N+1}$ are twice $\lambda_{\min} =N+1$ and $N-1$ times $\lambda_{\max} =2N+1$. This follows from analyzing the null-space of the matrix $\mathbf{C}_{N+1}\mathbf{C}^T_{N+1} - \lambda_* I$, where $* \in \{\max, \min\}$. In the spectral norm, we have
$$||\mathbf{C}_{N+1}||_2^2  = \lambda_{\max} = 2N + 1$$
    and
    $$||\mathbf{C}_{N+1}^{-1}||_2^2 = \frac{1}{\lambda_{\min}} = \frac{1}{N+1}.$$
    Then the condition number is given by their product, and $$\kappa_2(\mathbf{C}_{N+1}) = \sqrt{ \frac{2N + 1}{N + 1} } < \sqrt{2}.$$
    Lastly, by the equivalence of norms \eqref{eq:cond_num_ineq} we conclude that
    $$\kappa_F(\mathbf{C}_{N+1}) < \sqrt{2}(N+1).$$
\end{proof}

\subsubsection{Case \texorpdfstring{$n = 2^r p^s$}{n = 2\^r p\^s}}

Let $p$ be an odd prime number, and $n = 2^r p^s$ a composite number with $r\geq 2$, $s \geq 1$. Fix $N = 2^{r-2}p^{s} - 1$, and consider the grid of points 
$$x_i = 2\cos\left(\frac{2\pi (2i + 1)}{n} \right).$$ 
For this $n$, we define the related cosine matrix as $\left(\mathbf{C}_{N+1}\right)_{i,j} =  c_{ij}$, where $i,j =0,1,\dots, N$ and $c_{ij} = V_j(x_i)$. Writing this in matrix form yields
\begin{align*}
    {\mathbf{C}_{N+1}} &=
    \begin{bmatrix}
        1 & 2\cos\left( \frac{2\pi}{n} \right) & 2 \cos \left(\frac{2 \pi 2}{n}\right) & \hdots & 2\cos\left(\frac{2\pi N}{n}\right) \\
        \vdots & \vdots & \vdots & \hdots  & \vdots \\
        1 & 2\cos\left(\frac{2\pi (2i+1)}{n}\right) & 2 \cos\left(\frac{2 \pi 2 (2i +1)}{n}\right) & \hdots & 2\cos\left(\frac{2\pi \sigma (2i +1)N}{n}\right) \\
        \vdots & \vdots & \vdots & \ddots & \vdots \\
        1 & 2\cos\left(\frac{2\pi(2N+1)}{n}\right) & \hdots & \hdots & 2\cos\left(\frac{2\pi (2N+1)N}{n}\right)
    \end{bmatrix}.
\end{align*}

We apply the same strategy as above to derive a polynomial bound on the Frobenius norm of $\mathbf{C}_{N+1}$. We start with two lemmas.
\begin{lemma}
\label{lemma:orthogonal2}
    For any $1 \leq j \leq 2^{r-1}p^s-1$, we have the result
    \begin{equation} 
    \sum_{i=0}^{N} 2\cos \left( \frac{2\pi (2i + 1)}{n} j\right) = 0.
    \end{equation}
\end{lemma}

\begin{proof}
Let $\zeta = e^{\frac{2\pi j}{n}}$ be primitive $n$-th root of unity, and using the fact that 
$$\cos \left( \frac{2\pi j (2i +1)}{n}\right) = \cos\left( \frac{2\pi j (n - 2i - 1)}{n} \right)$$
the sum simplifies to
\begin{align*}
    \sum_{i=0}^{N} 2\cos \left( \frac{2\pi (2i +1) j}{n} \right) &=  \sum\limits_{i=0}^{2N+1} \cos \left( \frac{2\pi (2i+1) j}{n} \right) \\
    & =  \Re\left( \zeta + \zeta^3 + \dots+ \zeta^{n-1} \right)\\
    & =  \Re \left(\zeta \frac{\zeta^{n}-1}{\zeta^2 - 1} \right)  \\
    & =  0.
\end{align*}
\end{proof}

\begin{lemma}
\label{lemma:dct2rps}
Let $p$ be a prime number, and $\mathbf{C}_{N+1}$ the cosine matrix of dimension $N+1 = 2^{r-2}p^s$. Then the columns of $\mathbf{C}_{N+1}$ are orthogonal and
\begin{equation}
\mathbf{C}_{N+1}^T \mathbf{C}_{N+1} = \begin{pmatrix}
    N +1& 0 & \dots & 0 & 0\\
    0 & 2(N+1) & \dots & 0 & 0 \\
    \vdots & \vdots & \ddots & \dots& \vdots \\
    0 & 0 & \dots & 2(N +1) & 0 \\
    0 & 0 & \dots& 0 & 2(N+1)
\end{pmatrix}.
\end{equation}
\end{lemma}

\begin{proof}
First, consider $u,v > 1$. For these we have 
\begin{align*}
    (\mathbf{C}_{N+1}^T \mathbf{C}_{N+1})_{uv} &= 4\sum_{i=0}^{N} \cos\left(\frac{2\pi (2i +1)}{n} u\right) \cos\left(\frac{2\pi (2i +1)}{n} v\right)   \\
    &= 2\sum_{i=0}^{N} \cos\left(\frac{ 2\pi(2i +1)(u+v)}{n}\right ) + \cos\left(\frac{ 2\pi (2i+1)(u-v)}{n}\right ) \\
    &= \sum_{i=0}^{N} 2\cos\left(\frac{ 2\pi (2i + 1) (u-v)}{n}\right ).
\end{align*}

Now, if $u \neq v$ then again by Lemma \ref{lemma:orthogonal2},
$(\mathbf{C}_{N+1}^T \mathbf{C}_{N+1})_{uv} = 0$, otherwise $(\mathbf{C}_{N+1}^T \mathbf{C}_{N+1})_{uv} = 2N +2$. Finally, if $u=1$ and $v > 1$, then by Lemma \ref{lemma:orthogonal2} $(\mathbf{C}_{N+1}^T \mathbf{C}_{N+1})_{uv} = 0$, and for $u = v = 1$ we have $(\mathbf{C}_{N+1}^T \mathbf{C}_{N+1})_{uv} = N + 1$.
\end{proof}

The orthogonality result of Lemma \ref{lemma:dct2rps} allows us to easily prove the following bound.
\begin{lemma}
\label{lemma:cond_number_2rps}
    The condition number of the cosine matrix  $\mathbf{C}_{N+1}$ satisfies 
    \begin{align*}
        \kappa_F(\mathbf{C}_{N+1}) < \sqrt{2}(N+1).
    \end{align*}
\end{lemma}

\begin{proof}
    From Lemma \ref{lemma:dct2rps} we see that the eigenvalues of $\mathbf{C}_{N+1}^T\mathbf{C}_{N+1}$ are $N$-times $\lambda_{max} =  2(N+1)$ and once $\lambda_{\min} = N+1$. Therefore, the spectral norms are 
    $$||\mathbf{C}_{N+1}||_2^2 = \lambda_{\max} = 2(N+1)$$
    and
    $$||\mathbf{C}_{N+1}^{-1}||_2^2 = \frac{1}{\lambda_{\min}} = \frac{1}{N+1}.$$
    As before, the condition number is the product of the two,
    $$\kappa_2(\mathbf{C}_{N+1}) = \sqrt{\frac{\lambda_{\max}}{\lambda_{\min}}} = \sqrt{\frac{2(N+1)}{N+1}} = \sqrt{2}$$
    and by \eqref{eq:cond_num_ineq} we conclude that
    $$\kappa_F(A) < \sqrt{2}(N+1).$$
\end{proof}

\subsection{LWE and structured variants}
The \textit{Learning With Errors} paradigm is one of the most promising in the quest for quantum-safe classical cryptography. The security of LWE relies on lattice problems, and it consists of, loosely speaking, efficiently solving a linear system that has been tweaked by a random vector.

There exist a number of variants of this paradigm that differ in the
mathematical structure of the data used in the system. In purely LWE schemes, also referred to as \textit{unstructured} LWE variants, the mathematical structure is just the ring $\Z / q\Z$ of rational integers modulo $q$. 

For \textit{structured} LWE variants, we have \textit{Ring-LWE} and 
\textit{Polynomial-LWE}. In these paradigms, the mathematical structures from which the terms are drawn are rings of integers or polynomial quotient rings.
\begin{definition}[R/PLWE distributions]
Let $K$ be a number field and $\mathcal{O}_K$ be its ring of integers. Let $q$ be a rational prime, $f(x) \in \Z[x]$ a monic irreducible polynomial in $\Z[x]$, and $\mathcal{O}_f$ the associated quotient ring $\Z[x]/(f(x))$. Let $\chi$ be a discrete random distribution with values in $\mathcal{O}_K/q\mathcal{O}_K$ 
(resp. $\mathcal{O}_f/q\mathcal{O}_f$). For $s \in$ 
$\mathcal{O}_K/q\mathcal{O}_K$ (resp. $\mathcal{O}_f/q\mathcal{O}_f$), we define the \emph{(primal) RLWE (resp. PLWE) distribution} $\mathcal{A}_{s, \chi}$ 
\emph{(resp. $\mathcal{B}_{s, \chi}$)} as the distribution over 
$\mathcal{O}_K/q\mathcal{O}_K \times \mathcal{O}_K/q\mathcal{O}_K$ 
(resp. $\mathcal{O}_f/q\mathcal{O}_f \times \mathcal{O}_f/q\mathcal{O}_f$) obtained by sampling an element $a$ in $\mathcal{O}_K/q\mathcal{O}_K$ (resp. $\mathcal{O}_f/q\mathcal{O}_f$) uniformly at random, drawing an element $e$ according to $\chi$, and outputting the pair $(a, a\cdot s + e)$.
\end{definition}

\begin{definition}[R/PLWE problems]
Following the same notation as above, the two R/PLWE problems are as follows:

\emph{Search RLWE (resp. PLWE)} asks an adversary to return the secret $s$ with non-negligible probability when the adversary is given access to arbitrarily many samples of the RLWE (resp. PLWE) distribution.

\emph{Decision RLWE (resp. PLWE)} asks the adversary to decide whether a given random distribution is either uniform or the RLWE (resp. PLWE) distribution, with non-negligible probability when the adversary is given access to arbitrarily many samples of that given random distribution. 
\end{definition}

\section{Reduction formulas}
\label{sec:reduction_formulas}

It is well known that for any $n \in \Z_+$ and $m = \phi(n) / 2$, the elements in the set $S = \{\cos \left(2\pi j / n\right) :  j=0,1,\dots, m-1\}$ are linearly independent over $\Q$. However, the values $\cos \left(2\pi j / n\right)$ for $j \geq m$ are $\Q$-linearly dependent over $S$. In general, the dependency is true for $\cos \left(2\pi \sigma j / n \right)$ for any $\sigma$ coprime to $n$. We proceed to derive explicit formulas for the dependency relations in terms of the modified Chebyshev polynomials $V_0(x),V_1(x), \dots, V_{m-1}(x)$ from Definition \ref{def:modified_chebyshev_polynomials_V}. These formulas will be used later for matrix eliminations in the proof of the RLWE--PLWE equivalence and to find small degree representatives in the quotient ring $\Z[x] / (\Psi_n(x))$.

\subsection{Minimal polynomial of \texorpdfstring{$\psi_n$}{psi\_n}}
\label{sec:minimal_polynomial}
To derive formulas for reduction in the quotient ring $\Z[x] / (\Psi_n(x))$, we first need explicit formulas for the minimal polynomials $\Psi_n(x)$.
\begin{lemma}
\label{lemma:min_poly_prime}
    Let $p >3$ be a prime number and $m = \phi(p) / 2 = (p-1) / 2$. Then the minimal polynomial of $\psi_{p} = 2 \cos(2\pi / p)$ over $\Q$ is
    $$\Psi_p(x) = \sum_{i=0}^m V_i(x) = V_m(x) + V_{m-1}(x) + \dots + V_1(x) + V_0(x).$$
\end{lemma}

\begin{proof}
    From the symmetry of the cosines 
    $$
     \cos \left( \frac{2 \pi (m+1)}{p} \right) = \cos \left( \frac{2 \pi m}{p} \right),
    $$
    we conclude that $\psi_p$ is a root of $P(x) = V_{m+1}(x) - V_{m}(x)$ and then $\Psi_p(x) | P(x)$. Also, since $V_k(2) = 2$ for every $k \geq 1$, we find that $P(2) = 0$ and $x-2 | P(x)$. Finally, since $\Psi_p(x) | \frac{P(x)}{x-2}$ and both polynomials are monic of degree $m$, we conclude that
    \begin{align*}
        \Psi_p(x) &=  \frac{P(x)}{x-2} \\
        &= \frac{ V_{m+1}(x) - V_{m}(x)}{V_1(x) - 2} \\
        &=  V_{m}(x) + V_{m-1}(x) + \dots + V_1(x) + V_0(x).
    \end{align*}
    The last step follows from expanding the product 
    $$\left(V_{m}(x) + V_{m-1}(x) + \dots + V_1(x) + V_0(x)\right)(V_1(x)-2)$$
    and using the relation $V_1(x)V_j(x) = V_{j+1}(x) + V_{j-1}(x)$ from Proposition \ref{prop:Vn_times_Vm}, which leads to a telescoping sum.
\end{proof}

\begin{corollary}
\label{cor:min_poly_of_p_to_s}
    Let $p > 3$ be a prime number, $s \geq 1$, and $k = (p-1) / 2$. Then the minimal polynomial of $\psi_{p^s}$ is
\begin{align*}
\Psi_{p^s}(x) & = V_{kp^{s-1}}(x) + V_{(k-1)p^{s-1}}(x) + \dots + V_{p^{s-1}}(x) + V_0(x)\\
& = \sum_{i =0}^{k} V_{ip^{s-1}}(x).    
\end{align*}
\end{corollary}

\begin{proof}
    From \cite[Theorem 2.6 (CR3)]{loper2016resultants} and Lemma \ref{lemma:min_poly_prime} we obtain 
    \begin{align*}
     \Psi_{p^s}(x) &= \Psi_p(V_{p^{s-1}}(x))
     = \sum_{i=0}^m V_i(V_{p^{s-1}}(x)) = \sum_{i=0}^m V_{i p^{s-1}}(x).
    \end{align*}
\end{proof}

\begin{lemma} 
\label{lemma:min_poly_2powerp}
    Let $p >3$ be a prime number and $k = (p-1) / 2$. Then
    $$\Psi_{2^r p}(x) = \sum_{i=0}^k (-1)^{k - i}V_{i2^{r-1}}(x).$$
\end{lemma}

\begin{proof}
From \cite[Theorem 2.6 (CR1) and (CR4)]{loper2016resultants}, we obtain 
\begin{align*}
 \Psi_{2^r p}(x) = \frac{\Psi_{2^r}(V_p(x))}{\Psi_{2^r}(V_1(x))} 
 = \frac{V_{p 2^{r-2}}(x)}{V_{2^{r-2}}(x)}.
\end{align*}
On the other hand, notice that
\begin{align*}
    V_{2^{r-2}}(x) \sum_{i=0}^k (-1)^{k-i} V_{i 2^{r-1}}(x)  &= \sum_{i=0}^k (-1)^{k-i} V_{2^{r-2}}(x) V_{i 2^{r-1}}(x) \\
    &= (-1)^k V_{2^{r-2}}(x)V_{0}(x) \,  \\
    &+\sum_{i=1}^k (-1)^{k-i} \left( V_{(2i + 1)2^{r-2}}(x) + V_{(2i-1)2^{r-2}}(x)\right)\\
    &= (-1)^{k}V_{2^{r-2}}(x)  \\
    &+ \left( V_{(2k+1)2^{r-2}}(x) + (-1)^{k-1}V_{2^{r-2}}(x) \right) \\
    &= V_{(2k+1)2^{r-2}}(x) \\
    &= V_{p2^{r-2}}(x).    
\end{align*}
Now, reordering the equation gives 
$$ \frac{V_{p2^{r-2}}(x)}{V_{2^{r-2}}(x)} =\sum_{i=0}^k (-1)^{k-i} V_{i 2^{r-1}}(x).$$
\end{proof}

\begin{corollary}\label{cor:2rps}
    Let $p$ be an odd prime, $n = 2^r p^s$ and $k = (p-1) / 2$. Furthermore, denote by $$m = \frac{\phi(n)}{2} = (p-1) 2^{r-2}p^{s-1}$$
    the degree of the number field $\Q(\psi_n)$. Then the minimal polynomial of $\psi_n$ is
    $$\Psi_{n}(x) = \sum_{i =0}^{k} (-1)^{k-i} V_{i 2^{r-1}p^{s-1}}(x)$$
\end{corollary}

\begin{proof}
    From \cite[Theorem 2.6 (CR1)]{loper2016resultants} and Lemma \ref{lemma:min_poly_2powerp} we obtain 
\begin{align*}
 \Psi_{2^r p^s}(x) &= \Psi_{2^r p}(V_{p^{s-1}}(x)) \\
 &= \sum_{i =0}^{k} (-1)^{k-i} V_{i 2^{s-1}p^{s-1}}(x).
\end{align*} 
\end{proof}

\subsection{Modular reduction}
With the formulas for the minimal polynomials in hand, we are ready to express the dependency relations of the elements $\cos(2\pi j / n)$. We will study the two cases, $n = p^s$ and $n = 2^r p^s$, separately.

\subsubsection{Case \texorpdfstring{$n = p^s$}{n = p\^s}}
Let $N = (p^s - 1) / 2$, $k = (p-1) / 2$ and $m = \phi(p^s) / 2 = p^{s-1}k$. Therefore, the minimal polynomial of $\psi_n = 2 \cos(2 \pi / n)$ over $\Q$ is
\begin{align*}
    \Psi_{n}(x) &= \sum_{j = 0}^{k} V_{j p^{s-1}}(x)
\end{align*}
as given by Corollary \ref{cor:min_poly_of_p_to_s}.
\begin{lemma}
\label{lemma:equivalence_p_to_s}
    In the quotient ring $\Z[x] / (\Psi_{n}(x))$ we have the following identities 
    \begin{align}
    \label{eq:Vm_in_quotient}
        V_m(x) = &- \sum_{j=0}^{k-1} V_{j p^{s-1}}(x), \\
    \label{eq:V_m+l_in_quotient}
        V_{m+l}(x) =&  - \sum_{j=0}^{k-1} V_{j p^{s-1} + l}(x) - \sum_{j =1}^{k}V_{j p^{s-1} - l} (x).
    \end{align} 
\end{lemma}

\begin{proof}

In the ring $\Z[x] / (\Psi_{n}(x))$ we have 
\begin{align*}
    \Psi_n(x) = \sum_{j =0}^{k} V_{jp^{s-1}}(x) = 0.
\end{align*}
Therefore,
$$ 0 = V_m(x) + \sum_{j =0}^{k-1} V_{jp^{s-1}}(x),$$
and we conclude \eqref{eq:Vm_in_quotient}.

Additionally, for any $l$ satisfying $1 \leq l \leq N - m = (p^{s-1} - 1)/2 < m$, we have the relation
\begin{align*}
    V_m(x)V_l(x) &= V_{m + l}(x) + V_{m-l}(x)
\end{align*}
between the modified Chebyshev polynomials. Now, by multiplying the relation \eqref{eq:Vm_in_quotient} by $V_l(x)$ on both sides we find that
\begin{align*}
    V_{m+l}(x) + V_{m-l}(x) &= - V_l(x) -\sum_{j = 1}^{k-1} V_{j p^{s-1} + l}(x) + V_{j p^{s-1} - l}(x).
\end{align*}
Thus, we reorder the terms to
\begin{align*}
    V_{m+l}(x) &= - V_l(x) - V_{m-l}(x) - \sum_{j = 1}^{k-1} V_{j p^{s-1} + l}(x) + V_{j p^{s-1} - l}(x) \\
    & = - \sum_{j=0}^{k-1} V_{j p^{s-1} + l}(x) - \sum_{j =1}^{k}V_{j p^{s-1} - l} (x),
\end{align*}
which concludes the proof for \eqref{eq:V_m+l_in_quotient}.
\end{proof}

Note that under the isomorphism of rings $\Z[x] / (\Psi_{n}(x)) \cong  \Z[\psi_n]$ via the evaluation map $x \mapsto \psi_n$, we discover the relations
\begin{align*}
    V_{m+l}(\psi_n) &= - \sum_{j=0}^{k-1} V_{j p^{s-1} + l}(\psi_n) - \sum_{j =1}^{k}V_{j p^{s-1} - l} (\psi_n).
\end{align*}
Moreover, we have the isomorphisms $\sigma \in \text{Gal}(\Q(\psi_n) / \Q)$ giving the relation
\begin{align*}
    V_{m+l}(\psi_{\sigma,n}) &= - \sum_{j=0}^{k-1} V_{j p^{s-1} + l}(\psi_{\sigma,n}) - \sum_{j =1}^{k}V_{j p^{s-1} - l} (\psi_{\sigma,n}).
\end{align*}
Finally, by using the property $V_j(2 \cos(\theta)) = 2 \cos(j \theta)$ we have derived the formula
\begin{align}
\label{eq:reduction_formula_for_column_elimination}
    2 \cos \left(\frac{2 \pi \sigma (m+l)}{n}\right) &+ \sum_{j = 0}^{k-1} 2 \cos\left(\frac{2 \pi \sigma (j p^{s-1} + l)}{n}\right)  \nonumber \\
    &+ \sum_{j = 1}^{k} 2 \cos\left( \frac{2 \pi \sigma (j p^{s-1}-l)}{n}\right) = 0. 
\end{align}
This is the main equation that we will use in Section \ref{sec:plwe_rlwe_equivalence} to perform column elimination on a certain cosine matrix. Note that for the given bounds on $l$ and $j$, all the indices $jp^{s-1} + l$ and $jp^{s-1} - l$ are distinct.

\subsubsection{Case \texorpdfstring{$n=2^r p^s$}{n = 2\^r p\^s}}

Let $N = 2^{r-2}p^s - 1 $, $k = (p-1) / 2$ and $m = \phi(n) / 2 = k 2^{r-1}p^{s-1}$. From Corollary \ref{cor:2rps}, we know that the minimal polynomial of $\psi_n$ equals
\begin{align*}
\Psi_{2^r p^s}(x) &= \sum_{j =0}^{k} (-1)^{k-j} V_{j 2^{r-1}p^{s-1}}(x).
\end{align*}

\begin{lemma}
\label{lemma:2rps_reduction_relations}
    In the quotient ring $\Z[x] / (\Psi_{n}(x))$ we have the following identities 
    \begin{align}
    \label{eq:Vm_reduction_pow2}
        V_m(x) &= \sum_{j=0}^{k-1} (-1)^{k-j-1} V_{j 2^{r-1}p^{s-1}}(x) \\
    \label{eq:V_m+l_reduction_pow2}
        V_{m+l}(x) &= \sum_{j=0}^{k-1} (-1)^{k-1-j} V_{j2^{r-1} p^{s-1} + l}(x) + \sum_{j =1}^{k} (-1)^{k+1-j}V_{j 2^{r-1}p^{s-1} - l}(x).
    \end{align} 
\end{lemma}

\begin{proof}
In the quotient $\Z[x] / (\Psi_{n}(x))$, we see that
\begin{align*}
    \Psi_n(x) = \sum_{j =0}^{k} (-1)^{k-j} V_{j 2^{r-1}p^{s-1}}(x) = 0,
\end{align*}
and \eqref{eq:Vm_reduction_pow2} is a rearrangement of this sum.
Moreover, for any index $l$ such that $1 \leq l \leq N - m = 2^{r-2}p^{s-1} - 1 < m$, we have as above the formula
\begin{align*}
    V_m(x)V_l(x) &= V_{m + l}(x) + V_{m-l}(x).
\end{align*}
Multiplying the equation \eqref{eq:Vm_reduction_pow2} by $V_l(x)$ on both sides leads to the form
\begin{align*}
    V_{m+l}(x) &= \sum_{j=0}^{k-1} (-1)^{k-1-j} V_{j 2^{r-1}p^{s-1} + l}(x) + \sum_{j =1}^{k} (-1)^{k+1-j}V_{j 2^{r-1}p^{s-1} - l} (x).
\end{align*}
\end{proof}
Similarly to the case $n=p^s$, the Galois maps $\sigma \in \text{Gal}(\Q(\psi_n) / \Q)$ and the evaluation map $x \mapsto \psi_n$ endows us with the linear dependence relation between the cosine terms of the form $\cos \left(2 \pi \sigma (m+l) / n \right)$ and the set of cosines
\begin{align*}
    \cos\left(\frac{2 \pi \sigma (j 2^{r-1}p^{s-1} \pm l)}{n}\right), \quad \text{where }  0 \leq j \leq k, \,  1 \leq l \leq N-m.
\end{align*}
This relation will play an essential role in the proof of the PLWE--RLWE equivalence of the maximal real subextension of $\Q(\zeta_n)$.

\section{PLWE--RLWE equivalence}
\label{sec:plwe_rlwe_equivalence}
To show the equivalence of the RLWE and PLWE problems, the approach we take is identical to \cite{ahola2024fast}. Fix an odd prime $p$. Let $n = p^s$ or $n = 2^r p^s$ for $r \geq 1$ and $s \geq 1$. Further, let $m = \phi(n) / 2$. As stated in Section \ref{sec:preliminaries}, the set 
$$V := \{V_0(x), V_1(x), \hdots, V_{m-1}(x)\}$$ 
of the modified Chebyshev polynomials is a basis for $\Z[x] / (\Psi_n(x))$. Thus, the canonical embedding $\mathcal{M}: \Z[x] / (\Psi_n(x)) \to \R^m$ is given by
\begin{align*}
    a_0 V_0(x) + a_1 V_1(x) + \hdots + a_{m-1} V_{m-1}(x) &\mapsto \mathbf{M} (a_0, \, a_1, \hdots, \ a_{m-1})^T,
\end{align*}
where
\begin{align*}
    \underset{m \times m}{\mathbf{M}} &=
    \begin{bmatrix}
        1 & 2\cos(2\pi / n) & 2 \cos(2 \pi 2 / n) & \hdots & 2\cos(2\pi (m-1) / n) \\
        \vdots & \vdots & \vdots & \hdots  & \vdots \\
        1 & 2\cos(2\pi \sigma / n) & 2 \cos(2 \pi \sigma 2 / n) & \hdots & 2\cos(2\pi \sigma (m-1) /n) \\
        \vdots & \vdots & \vdots & \ddots & \vdots \\
        1  & \hdots & \hdots & \hdots & 
    \end{bmatrix}
\end{align*}
is an $m$-by-$m$ matrix with $\sigma \in \{1,2,\hdots, n/2\}$ and $(\sigma, n) = 1$. 

From \cite{RSW:2018:RPP} and \cite{DD:2017}, we know that the PLWE and RLWE problems are equivalent, if the condition number of the matrix $\mathbf{M}$ is bounded by a polynomial in $n$. For a short summary of this approach, see \cite[Definition 2.3]{Blanco:2021:VMTR}. With this setup, we are ready to state and prove the two main theorems of this section, one for $n = p^s$ and another for $n = 2^r p^s$.
\begin{theorem}
\label{thm:PLWE_RLWE_equivalence_p_to_s}
    Let $p \geq 3$ be a prime and $s \geq 1$. Then PLWE and RLWE are equivalent for the maximal real subextension of the $p^s$-th cyclotomic field.
\end{theorem}

\begin{proof}
Let $N = (p^s - 1) / 2$. Our goal is to bound the Frobenius norm of $\mathbf{M}$ and its inverse $\mathbf{M}^{-1}$ by using the bound from Lemma \ref{lemma:cond_number_of_C_N+1} on the condition number of the cosine matrix $\mathbf{C}_{N+1}$ from Definition \ref{def:cosine_matrix_p_to_s}.
First of all, notice that there exist a permutation of the rows of $\mathbf{C}_{N+1}$ and hence a permutation matrix $\mathbf{P}$ such that
\begin{align*}
    \mathbf{P} \mathbf{C}_{N+1} &=
    \begin{bmatrix*}
        \underset{m \times m}{\mathbf{M}} & \underset{m \times m'}{\mathbf{B}} \\ \underset{m' \times m}{\mathbf{A}} & \underset{m' \times m'}{\mathbf{C}}
    \end{bmatrix*},
\end{align*}
where $m' = N+1 - m = (p^{s-1} + 1) / 2$, and the submatrices $\mathbf{A}, \mathbf{C}$ and $\mathbf{C}$ have the indicated dimensions. In particular, the entries of the matrix $\mathbf{B}$ are of the form $2\cos (2\pi \sigma j / n)$ with $\sigma$ coprime to $p$ and the column index $j = m, m+1,\dots, N$. As mentioned in Section \ref{sec:reduction_formulas}, for a fixed $\sigma$ the corresponding cosine values are linearly dependent over the set $S_\sigma = \{2\cos (2\pi \sigma j / n) : j = 0,1,\hdots,m-1\}$. In particular, the columns of $\mathbf{B}$ are linearly dependent of the columns of $\mathbf{M}$, and the general dependency relations are given by Lemma \ref{lemma:equivalence_p_to_s} as formulas \eqref{eq:Vm_in_quotient} and \eqref{eq:V_m+l_in_quotient}. The explicit formula for the different values of $\sigma$ and $l$ is \eqref{eq:reduction_formula_for_column_elimination}.

Let $(\mathbf{F})_{i,l} = F_{i,l}$ be the $m\times m'$ matrix indexed by $i=0,1,\dots, m-1$ and $l=0,1,\dots, N-m$. The matrix $\mathbf{F}$ is related with the dependency relations \eqref{eq:Vm_in_quotient} and  \eqref{eq:reduction_formula_for_column_elimination} as follows:

The first column of $\mathbf{F}$ with $l=0$ is related with the dependency of $V_m(x)$ in formula \eqref{eq:Vm_in_quotient} as 
$$F_{i,0} = \left\{ \begin{array}{cl}
     1, &  p^{s-1} \mid  i  \\
     0, & \text{otherwise}.
\end{array} \right. $$
For the columns $l$ with $1 \leq l \leq N-m$, we use the dependency formula  of $V_{m + l}(x)$, \eqref{eq:reduction_formula_for_column_elimination} as
$$F_{i,l} = \left\{ \begin{array}{cl}
     1, &  i \equiv \pm l\mod p^{s-1}  \\
     0, & \text{otherwise}. 
\end{array} \right.$$
Finally, we notice by counting the number of ones in $\mathbf{F}$ that
$$||\mathbf{F}||_F^2 = \frac{p-1}{2} + 2\left(\frac{p-1}{2}\right)(N-m) = m \leq N.$$

Then we can eliminate the matrix $\mathbf{B}$ by performing column operations:
\begin{align}
\label{eq:C_N+1_columns_reduced}
    \mathbf{P} \mathbf{C}_{N+1} \mathbf{R} = &
    \begin{bmatrix*}
        \underset{m \times m}{\mathbf{M}} & \underset{m \times m'}{\mathbf{0}} \\ \underset{m' \times m}{\mathbf{A}} & \underset{m' \times m'}{\mathbf{D}}
    \end{bmatrix*},
\end{align}
where $\mathbf{R}$ is the $(N+1)\times(N+1)$ matrix given by
\begin{align*}
    {\mathbf{R}} &= 
    \begin{bmatrix*}
        \underset{m \times m}{\mathbf{I}} & \underset{m \times m'}{\mathbf{F}} \\ 
        \underset{m' \times m}{\mathbf{0}} & \underset{m' \times m'}{\mathbf{I}}
    \end{bmatrix*}.
\end{align*}

Notice that $\mathbf{M}$ is a submatrix of $\mathbf{P} \mathbf{C}_{N+1}$, so we get a strict bound on the condition number of the transform,
\begin{align*}
    ||\mathbf{M}||_F^2 < || \mathbf{P}\mathbf{C}_{N+1} ||_F^2 = ||\mathbf{C}_{N+1} ||_F^2.
\end{align*}
To bound the norm of the inverse $\mathbf{M}^{-1}$, we invert the matrix equation \eqref{eq:C_N+1_columns_reduced} to get
\begin{align}
    \mathbf{R}^{-1} \mathbf{C}_{N+1}^{-1} \mathbf{P}^{T} = &
    \begin{bmatrix*}
        \underset{m \times m}{\mathbf{M}^{-1}} & \underset{m \times m'}{\mathbf{0}} \\ \underset{m' \times m}{\mathbf{A'}} & \underset{m' \times m'}{\mathbf{D}^{-1}}
    \end{bmatrix*},
\end{align}
where the inverse of $\mathbf{R}$ is given by blockwise by
\begin{align*}
    {\mathbf{R}^{-1}} &= 
    \begin{bmatrix*}
        \underset{m \times m}{\mathbf{I}} & \underset{m \times m'}{-\mathbf{F}} \\ 
        \underset{m' \times m}{\mathbf{0}} & \underset{m' \times m'}{\mathbf{I}}
    \end{bmatrix*}.
\end{align*}

We use the sub-multiplicativity of the Frobenius norm to deduce that
\begin{align*}
    ||\mathbf{M}^{-1}||_F^2 &< || \mathbf{R}^{-1}\mathbf{C}_{N+1}^{-1} \mathbf{P}^T ||_F^2 \leq || \mathbf{R}^{-1} ||_F^2 ||\mathbf{C}_{N+1}^{-1} \mathbf{P}^T ||_F^2 \\
    &= || \mathbf{R}^{-1} ||_F^2 ||\mathbf{C}_{N+1}^{-1}||_F^2 = (N + 1 + ||\mathbf{F}||_F^2) ||\mathbf{C}_{N+1}^{-1}||_F^2  < 3N ||\mathbf{C}_{N+1}^{-1}||_F^2.
\end{align*}

By combining these two bounds, we get an upper bound for the condition number of $\mathbf{M}$, 
\begin{align*}
    \kappa_F(\mathbf{M})^2 &= ||\mathbf{M}^{-1}||_F^2 ||\mathbf{M}||_F^2
    < 3N ||\mathbf{C}_{N+1}^{-1}||_F^2 ||\mathbf{C}_{N+1}||_F^2 = 3N \kappa_F(\mathbf{C}_{N+1})^2.
\end{align*}
The final step is to use Lemma \ref{lemma:cond_number_of_C_N+1} which states that $\kappa_F(\mathbf{C}_N)^2 = \mathcal{O}(N^2)$. Putting everything together yields the desired bound 
\begin{align*}
    \kappa_F(\mathbf{M})^2 = 3N \mathcal{O}(N^2) = \mathcal{O}(N^3) = \mathcal{O}(n^3).
\end{align*}
This concludes the proof.
\end{proof}

\begin{theorem}
\label{thm:PLWE_RLWE_equivalence_p_to_s_pow_2}
    Let $p$ be an odd prime, $r \geq 2$, and $s \geq 1$. Then PLWE and RLWE are equivalent for the maximal real subextension of the $2^rp^s$-th cyclotomic field.
\end{theorem}

\begin{proof}
Let $N = 2^{r-2}p^s - 1$, $n = 2^r p^s$ and $m = \phi(n) / 2 = 2^{r-1}p^{s-1}k$, where $k = (p-1)/2$. Let $\mathbf{C}_{N+1}$ be the cosine matrix for $n = 2^r p^s$ from Section \ref{sec:preliminaries}. As in the proof of Theorem \ref{thm:PLWE_RLWE_equivalence_p_to_s}, there exists a matrix $\mathbf{P}$, a permutation of rows, such that 
\begin{align*}
    \mathbf{P} \mathbf{C}_{N+1} &=
    \begin{bmatrix*}
        \underset{m \times m}{\mathbf{M}} & \underset{m \times m'}{\mathbf{B}} \\ \underset{m' \times m}{\mathbf{A}} & \underset{m' \times m'}{\mathbf{C}}
    \end{bmatrix*},
\end{align*}
where $m' = N + 1 - m = 2^{r-2}p^{s-1}$, and all the matrices have the indicated dimensions. Similarly, the entries of the matrix $\mathbf{B}$ are of the form $2\cos (2\pi \sigma j / n)$ with $\sigma$ coprime to $2p$ and $j = m, m+1,\dots, N$. By Lemma \ref{lemma:2rps_reduction_relations}, the columns of $\mathbf{B}$ are linearly dependent on the columns of $\mathbf{M}$. 

We continue to construct the column elimination matrix $\mathbf{F}$ of dimension $m \times m'$. Denote by $q(i)$ the quotient of $i$ divided by $2^{r-1}
p^{s-1}$, that is, 
$$q(i)= \frac{i}{2^{r-1}p^{s-1}}.$$
Set $(\mathbf{F})_{i,l} = F_{i,l}$ indexed by $i = 0,1, \hdots, m-1$ and $l = 0,1, \hdots, N-m$ and define the entries of $\mathbf{F}$ by
$$F_{i,0} = \left\{ \begin{array}{rl}
    (-1)^{k - q(i)}, & \text{if }2^{r-1}p^{s-1} \mid i \\
     0, & \text{otherwise}.
\end{array} \right.$$
For the columns $l$ with $1 \leq l \leq N-m$,
$$F_{i,l} = \left\{ \begin{array}{rl}
   (-1)^{k - q(i-l)},  &  \text{if } i \equiv l \mod 2^{r-1}p^{s-1} \\
   (-1)^{k - q(i + l)},  & \text{if } i \equiv - l \mod 2^{r-1}p^{s-1} \\
    0, & \text{otherwise}.
\end{array} \right.$$

Therefore, 
$$|| \mathbf{F}||_F^2 = \frac{p-1}{2} + 2\left(\frac{p-1}{2}\right)(N-m) = \frac{p-1}{2}(2^{r-1}p^{s-1} -1) < N.$$

From here onwards, we follow the same steps as in the proof of Theorem \ref{thm:PLWE_RLWE_equivalence_p_to_s}, concluding  
\begin{align*}
    \kappa_F(\mathbf{M})^2 = \mathcal{O}(N^3) = \mathcal{O}(n^3).
\end{align*}
This finishes the proof.
\end{proof}

\section{Fast multiplication in the quotient ring}
\label{sec:fast_multiplication_via_DCT}
This section introduces an algorithm for fast multiplication over the polynomial quotient ring $\Z[x]/(\Psi_n(x))$ for $n = 2^r p^s$ with $p \geq 3$ a prime, $r>2$ and $s\geq 1$. The special cases $n=2^r$ and $n=2^r 3^s$ were considered in \cite{ahola2024fast}. The algorithm consists of two steps: first, fast multiplication in $\Z[x]$ via the {\em Discrete Cosine Transform} $\mathsf{DCT}$, followed by a fast reduction of the product modulo $\Psi_{n}$. The algorithm resembles the \emph{Number Theoretic Transform} (NTT). However, the NTT is typically computed over the cyclotomic ring $\Z[x]/(x^n + 1)$, where $n$ is a power of two \cite{agarwal1975ntt, pollard1971fast}. Non-power-of-two cyclotomic fields were studied in \cite{lyubashevsky2019nttru}, where the authors develop a fast multiplication algorithm in the quotient $\Z[x] / (f(x))$, where the polynomial modulus is $f(x) = x^n - x^{n / 2} + 1$ with $n = \phi(2^r 3^s)$, $r \geq 1$ and $s \geq 1$.

\subsection{Fast polynomial multiplication via the DCT}
Despite the fact that there are several fast algorithms for polynomial multiplications based on the {\em Fast Fourier Transform} $\mathsf{FFT}$, we are interested in a fast algorithm that multiplies two polynomials expressed in the $V$-basis \eqref{eq:V_basis}. The traditional methods based on the $\mathsf{FFT}$ perform the computations in the canonical power basis $\{1, x,\dots,x^{N-1}\}$. Our variant makes use of the {\em Discrete Cosine Transform} and its inverse.

The computation of the $\mathsf{DCT}$ and $\mathsf{IDCT}$ can be achieved through various algorithmic approaches, broadly categorized as indirect and direct methods. Indirect methods cleverly repurpose existing fast transformations. For instance, some early approaches, described in Ahmed et al. \cite{ahmed1974discrete}, utilize the $\mathsf{FFT}$ or Hadamard transforms to compute the $\mathsf{DCT}$. While being asymptotically effective, these indirect methods may not be optimal in terms of computational cost. 

In contrast, direct algorithms are specifically designed for the $\mathsf{DCT}$ and $\mathsf{IDCT}$. These algorithms often employ techniques like matrix factorization and recursive decomposition, drawing inspiration from the Cooley-Tukey algorithm (the foundational algorithm for the $\mathsf{FFT}$). Works by Bi and Yu, \cite{bi1998dct}, Bi \cite{bi1999fast}, Hou \cite{hou1987fast}, and Kok \cite{kok1997fast} demonstrate that direct $\mathsf{DCT}$ algorithms can achieve the same asymptotic time complexity as the $\mathsf{FFT}$, namely $\mathcal{O}(n \log n)$, through careful optimization. These direct approaches aim to minimize the number of arithmetic operations required, leading to improved practical performance.

\subsubsection{The Discrete Cosine Transform \texorpdfstring{$\mathsf{DCT}$}{DCT}}
The Discrete Cosine Transform ($\mathsf{DCT}$) is widely used in many
digital signal processing applications. Multiple fast algorithms for different types of the Discrete Cosine Transform have been reported in the literature \cite{bi1998dct, hou1987fast, kok1997fast, ahmed1974discrete}. In this paper, we shall use types II and III, which are mutual inverses when scaled properly.

\begin{definition}[$\mathsf{DCT}$]
Let $N \in \Z^+$ and $a(k), \, k= 0,1,\dots, N-1$ be a real sequence. The \emph{non-scaled type-III Discrete Cosine Transform of} $a(k)$ is the sequence $\hat{a}(j)$ defined by 
$$\hat{a}(j) = \frac{a(0)}{2} + \sum_{i=1}^{N-1} a(i)\cos\left( \frac{2\pi(2j +1) i}{4N}\right),  \quad 0\leq j \leq N-1.$$

The \emph{inverse of the type-III $\mathsf{DCT}$} is given by the type-II $\mathsf{DCT}$. The non-scaled type-II Discrete Cosine Transform of the sequence $a(k)$ is a new sequence $\tilde{a}(j)$ given by
$$ \tilde{a}(j) =  \sum_{i=0}^{N-1} a(i)\cos \left( \frac{2\pi(2i+1)j}{4N} \right), \quad 0\leq j \leq N-1.$$
\end{definition}

As mentioned before, the scaled type-III $\mathsf{DCT}$ and scaled type-II $\mathsf{DCT}$ are inverses of each other. For the non-scaled versions, we have a similar result:

\begin{lemma}\label{lemma dct} For any real sequence $a(k)$, $k=0,1,\dots, N-1$, we have 
$$\mathsf{IDCT} \left( \mathsf{DCT}(\mathbf{a}) \right) = \frac{N}{2}\mathbf{a},$$
where $\mathbf{a}$ denotes the vector $(a(0), \dots, a(N-1))^T$.
\end{lemma}

\subsubsection{Linear Convolution}

Let $p(x)$ be a polynomial of degree less than or equal to $N-1$, then $p(x)$ can be represented in basis $\{V_0(x), V_1(x), \dots, V_{N-1}(x) \}$ as  
$$p(x) = \sum_{i = 0}^{N-1} a_iV_i(x).$$

We will write $\mathbf{a} = (a_0,a_1, \dots,a_{N-1})^T$ for the vector of coefficients in the $V$-basis. Related to the cosine transform, we define a grid of points 
\begin{align}
    x_j := 2\cos\left(\frac{2\pi(2j + 1)}{4N}\right), \quad j = 0, 1,\hdots, N-1.
    \label{def:chebyshev_nodes}
\end{align}
The evaluation of $p(x)$ at the grid points $x_j$ yields
$$p(x_j) = a_0 + 2\sum_{i=1}^{N-1}a_i\cos\left( \frac{2\pi(2j+1)i}{4N} \right).$$
By writing this in vector form for all $j$, we have the property 
$$\mathbf{\hat{p}} = 2 \, \mathsf{DCT}(\mathbf{a}),$$
where $\mathbf{\hat{p}} =(p(x_0),p(x_1),\dots,p(x_{N-1}))^T$ is the vector of all the evaluations.

Finally, let $p(x), q(x) \in \Z[x] $ and $r(x) = p(x)q(x)$. By definition, the evaluations satisfy $r(x_j) = p(x_j)q(x_j)$ for all $x_j$ in the grid. Thus, the vectors of evaluations satisfy the property
$$\mathbf{\hat{r}} = \mathbf{\hat{p}} \odot \mathbf{\hat{q}},$$ 
where $\odot$ denotes the elementwise product of vectors. Written in another form,
\begin{equation}\label{conv_formula}
\mathsf{DCT}(\mathbf{c}) = 2 \, \mathsf{DCT}(\mathbf{a}) \odot \mathsf{DCT}(\mathbf{b}),
\end{equation}
where $\mathbf{a}, \mathbf{b}$ and $\mathbf{c}$ are the coefficient vectors of the polynomials $p,q$ and $r$ in the $V$-basis as given below.

Now, let $n = p^s$ with $s \geq 1$ or $n =2^r p^s$ with $r >2 $ and $ s\geq 1$. Further, let $m = \phi(n) / 2 = [\Q(\psi_n):\Q]$, and take two polynomials $p(x)$ and $q(x)$ in $\Z[x]/(\Psi_n(x))$. We denote also by $p(x)$ and $q(x)$ the coset representatives that have degree less or equal to $m-1$. Let $N$ be closest power of $2$ greater or equal to $2m$, and express $p(x)$ and $q(x)$ in basis $V = \{ V_0(x), \dots, V_{N-1}(x) \}$ as 
\begin{align*}
    p(x) = \sum_{i = 0}^{N-1} a_iV_i(x) \quad \text{and} \quad q(x) = \sum_{i = 0}^{N-1} b_iV_i(x).
\end{align*}
Clearly, $a_i = b_i = 0$ for $i \geq m$. If we denote the product by 
$$r(x) = p(x)q(x) = \sum_{i=1}^{N-1} c_i V_i(x),$$
then $r(x)$ is the only polynomial of degree less than $N-1$ satisfying the convolution formula $(\ref{conv_formula})$. Thus we can compute the coefficients of $r(x)$ by using the $\mathsf{DCT}$ and its inverse by
\begin{align*}
    \frac{N}{2} \mathbf{c}
    &= 2\, \mathsf{IDCT} \left(\mathsf{DCT}(\mathbf{a}) \odot \mathsf{DCT}(\mathbf{b}) \right),
\end{align*}
where $\mathbf{c} = (c_0,c_1,\dots,c_{N-1})^T$ is the coefficient vector of the polynomial $r(x)$ in $V$-basis. The previous formula implies that computing the coefficient of $p(x)q(x)$ in $V$-basis takes in total two $N$-point $\mathsf{DCT}$ transforms, one $N$-point $\mathsf{IDCT}$, and $N$ multiplications.

Each $N$-point $\mathsf{DCT}$ and its inverse $\mathsf{IDCT}$ requires exactly $(l-2)2^{l-2}$ multiplications and $3(l-2)2^{l-3} - 2^{l-2} + 1$  additions, where $2^l = N$, that is $l = \lceil \log_2(m) \rceil + 1$. The derivations for the number of operations can be found in \cite{kok1997fast, hou1987fast}. This means that the asymptotic complexity of computing the coefficients of the product is $\mathcal{O}(m \log m)$.

\subsection{Reduction modulo \texorpdfstring{$\Psi_n(x)$}{Psi\_n(x)}}
We proceed by describing dependency relations on the polynomial quotient ring $\Z[x] / (\Psi_{n}(x))$. In particular, we derive a formula for reducing the degree of polynomials modulo $\Psi_{n}(x)$.

We start by considering a polynomial $c(x)$ of degree $2m - 2$ in the $V$-basis
\begin{align*}
    c(x) &= \sum_{j = 0}^{2m - 2} c_j V_j(x).
\end{align*}
Note that every product of two polynomials of degree $m-1$ can be written in this form. We are interested in finding the reduction $\overline{c}(x) \equiv c(x) \mod \Psi_{n}(x)$ such that $\deg \overline{c}(x) < m$. In other words, we want to find a fast method for computing the coefficients $\overline{c}_j$ of the reduced polynomial
\begin{align*}
    \overline{c}(x) &= \sum_{j = 0}^{m-1} \overline{c}_j V_j(x).
\end{align*}

\subsubsection{Case \texorpdfstring{$n = p^s$}{n = p\^s}}
Let $n = p^s$, $N = (p^s-1)/2$ and $m = p^{s-1}(p-1)/2$. We have already derived the reduction formulas for $V_{j}(x)$ for the degrees $m \leq j \leq N$ in Section \ref{sec:reduction_formulas}. The formulas are the Equations \eqref{eq:Vm_in_quotient} and \eqref{eq:V_m+l_in_quotient} of Lemma \ref{lemma:equivalence_p_to_s}. However, we are yet to find a reduction formula for the indices $N < j \leq 2m - 2$. To this end, for the indices $N < j \leq 2m-2$, we have the following equivalences
\begin{align*}
    V_j(x) \equiv  V_{n - j}(x) \mod{\Psi_n(x)}.
\end{align*}
by the symmetry $\cos(2\pi - x) = \cos(x)$. Therefore, the reduction of a polynomial $c(x)$ of degree $2m - 2$ to a polynomial $\overline{c}(x)$ of degree $N$ takes $2m - 2 - N < m$ additions.

After these steps, we have found a polynomial $\overline{c}(x) \equiv c(x) \mod{\Psi_n(x)}$ with $\deg \overline{c}(x) = N$. The last step is to use the relations \eqref{eq:Vm_in_quotient} and \eqref{eq:V_m+l_in_quotient} to lower the degree of the representative below $m$, that is, to find a polynomial in the quotient of $c(x)$ expressed in the basis $\{V_0(x), V_1(x), \hdots, V_{m-1}(x)\}$.

By iterating backwards over the indices $j$ such that $m \leq j < N$ and accumulating the coefficients, we need $(N - m) 2 (p-1) / 2 < m$ operations. In total, the number of operations in the whole reduction process is bounded by $2m$. The asymptotic bound is $\mathcal{O}(m)$.

\subsubsection{Case \texorpdfstring{$n = 2^r p^s$}{n = 2\^r p\^s}}
Let $n = 2^r p^s$, $N = 2^{r-2}p^s-1 $ and $m = 2^{r-2}p^{s-1}(p-1)$. Section \ref{sec:reduction_formulas} presents reduction formulas for $V_j(x)$ modulo $\Psi_n(x)$ with indices $m \leq j \leq N$. Now, we introduce similar formulas for $N < j \leq 2m - 2$. First, note that
\begin{align*}
    V_{N + 1}(x) \equiv 0 \mod{\Psi_n(x)},
\end{align*}
since
\begin{align*}
    2 \cos \left( \frac{2 \pi (N + 1)}{n} \right) &= 2 \cos \left( \frac{2 \pi 2^{r-2}p^s}{2^r p^s} \right) = \cos\left(\frac{\pi}{2} \right) = 0.
\end{align*}
Second, due to the symmetry $2 \cos(\pi - x) = - 2\cos(x)$ of the cosines, we have the equivalence
\begin{align*}
    V_{2^{r-1}p^s - j}(x) \equiv - V_j(x) \mod{\Psi_n(x)}.
\end{align*}

With these relations the reduction of a polynomial $c(x)$ of degree $2m -2$ to a polynomial $\overline{c}(x) \equiv c(x) \mod{\Psi_n(x)}$ takes as before $2m - N - 2 < m$ operations, resulting in a polynomial of degree $N$.

Finally, we use the relations \eqref{eq:Vm_reduction_pow2} and \eqref{eq:V_m+l_reduction_pow2} to reduce the degree of $\overline{c}(x)$ to $m-1$. Thus, we have found a representative $\overline{c}(x)$ of $c(x)$ of degree $m-1$ using only $(N - m)(p-1) < m$ operations. Hence, the cost of the total reduction of $c(x)$ is $\mathcal{O}(m)$.

\subsection{Fast change of basis revisited}
The PLWE problem is typically stated using polynomials in the power basis. It follows that the security guarantees that follow from the RLWE--PLWE equivalence hold in principle only if we sample our polynomials in the power basis. However, as was shown in \cite{ahola2024fast}, there exists a fast change of basis to the $V$-basis and back. Their approach works for general polynomials of degree $N-1$. This allows us to compute the product $a(x) s(x) \in \Z[x] / (\Psi_n(x))$ with quasilinear complexity $\mathcal{O}(m \log m)$, even in the case of polynomials sampled from the power basis. 
\begin{theorem}
For $n = 2^r p^s$ with $r \geq 0$ and $s \geq 1$, given two polynomials $a(x), s(x) \in \Z[x] / (\Psi_n(x))$ in the power basis, their product $a(x) \cdot s(x) \in \Z[x] / (\Psi_n(x))$ can be computed with asymptotic complexity $\mathcal{O}(n \log n)$.
\end{theorem}
\begin{proof}
    The proof is a direct consequence of the general argument provided in \cite{pan1998new}. For a more detailed proof, see \cite[Section 5]{ahola2024fast}.
\end{proof}

\subsection{Computations over a finite field}
The Discrete Cosine Transform can be computed over a prime finite field $\F_q$ by writing the cosines in terms of elements in the finite field as
\begin{align*}
     \cos\left(\frac{2\pi(2j + 1)i}{M} \right) \mapsto 2^{-1} \left( \omega_{M}^{(2j + 1)i} + \omega_{M}^{-(2j + 1)i} \right),
\end{align*}
where $\omega_{M}$ is the $M$-th primitive root of unity in $\F_q$. For a more comprehensive discussion, see \cite{de2004discrete} . Here we need to choose $M = 4N$, where $N = N(n)$ is the size of the DCT. To guarantee the existence of $\omega_{M} \in \F_q$, we must choose $q \equiv 1 \mod M$.

Note that it is possible to use a composite modulus $pq$. First, we choose two primes $p \equiv q \equiv 1 \mod M$. We use the Chinese Remainder Theorem to construct an element $\omega \in \Z / pq \Z \cong \Z / p \Z \times \Z / q \Z$ such that both $\omega \mod p$ and $\omega \mod q$ have exact order $M$. Then we can perform the computations separately in the two finite fields and transfer the elements back to $\Z / pq \Z$ via the CRT isomorphism.

\section{Computational analysis of robustness against PLWE attacks}
\label{sec:computational_analysis}

Structured lattices and, in particular, structured variants of the \textit{Learning With Errors} paradigm have become one of the most promising solutions towards post-quantum cryptography. This statement is realized in the standardization selections of the First Standardization process of NIST. In July 2022, NIST announced the selection of 4 candidates to be standardized: 
\begin{itemize}
    \item Crystals-Kyber (i.e. ML-KEM \cite{FIPS:2024:203}), 
    \item Crystals-Dilithium (i.e. ML-DSA \cite{FIPS:2024:204}), 
    \item Falcon \cite{Falcon} (i.e. future FN-DSA),
    \item SPHINCS+ (i.e. SLH-DSA \cite{FIPS:2024:205}).
\end{itemize}

Out of the four selected schemes, three are based on some structured lattice paradigm. In particular, two of those selections are based on a structured variant of the LWE paradigm. Furthermore, other standardization competitions, like the Korean Post-Quantum Cryptography competition, have reached very similar results. A majority of the selected schemes have been based on structured lattices. 

The fact that an overwhelming proportion of cryptographic schemes are based on the 
structured lattice variants is due to their unique way of achieving, at the same time, acceptable performance and small cryptographic sizes for their associated elements. This follows from the additional algebraic structure. Note that in the case of unstructured LWE schemes, the size of the cryptographic keys is of 
order $\mathcal{O}(n^2)$, where $n$ is a security 
parameter which seeks to represent the desired strength of the cryptosystem. 

The algebraic structure of the structured LWE schemes decreases the size of cryptographic keys to linear order $\mathcal{O}(n)$. With regards to performance, the multiplication of elements within the structured setting can be made very efficient through a number of optimizations, such as the Toom-Cook or Karatsuba algorithms. The fast multiplication of elements is needed for a number of cryptographic subprocesses. For the power-of-two cyclotomic fields, the situation is optimal, since we can use the quasilinear Number Theoretic Transform (NTT) to perform the multiplications. The NTT can be seen as a special case of the Discrete Fourier Transform over finite fields.

\subsection{Root-based attacks against the additional structure}
The presence of additional algebraic structure can also be used as an additional source of attacks. Among them, the original work of Lauter et al. \cite{ELOS:2015:PWI, ELOS:2016:RCN} and the subsequent work \cite{BDM:2024:GARB} expanded the surface of instances that are vulnerable to algebraic attacks. The attacks exploit the information about the roots of the polynomial modulus $f(x)$ of the PLWE instance. The general approach of the attacks is as follows:
\begin{enumerate}
    \item Take a root $\alpha$ of the polynomial $f(x)$ and apply the evaluation homomorphism $x \mapsto \alpha$ to each of the samples $(a_i(x), b_i(x))$ to be distinguished.
    \item Loop through all possible values for $s(\alpha)$, that is, for the secret $s(x)$ under the evaluation homomorphism.
    \item For each guess for $s(\alpha)$, compute the associated error values $e_i(\alpha) = b_i(\alpha) - s(\alpha) \cdot a_i(\alpha)$.
    \item Perform certain distinguishability actions over the tentative evaluated errors to get a distinguishability feature.
\end{enumerate}
The particular actions taken on the tentative evaluated errors will define each of the attacks. We specify three possibilities:
\begin{enumerate}
    \item Analyze the set of possibilities for the error values, once evaluated in the root $\alpha$.
    \item Analyze the smallness of the evaluated error values when their image distribution is constrained to a certain interval.
    \item Analyze the smallness of the evaluated error values when their image distribution is not constrained to any particular interval.
\end{enumerate}

The success of all of the attacks above depends heavily on the order of certain elements of the polynomial. In \cite{ELOS:2015:PWI,ELOS:2016:RCN}, a number of attacks against the decisional PLWE are presented, assuming that the polynomial $f(x)$ has a root $\alpha$ of \textit{small} order in $\mathbb{F}_q$. In \cite{BDM:2024:GARB}, these attacks were generalized to roots belonging to arbitrary degree field extensions of $\mathbb{F}_q$, and the success of the attack depends on the polynomial $f(x)$ having certain \textit{k-ideal} factors of small order. By $k$-ideal factors, we mean factors of the form $x^k + a$, where the degree $k$ is small. The existence of these factors introduces a vulnerability if the negated constant term $-a$ has small order. 

Therefore, when considering the use of maximal real cyclotomic polynomials $\Psi_n(x)$ in the PLWE setting, it is important to analyze how effective these attacks are. In other words, we must investigate whether the maximal real cyclotomic polynomials are more prone to these important algebraic attacks than other better-known and well-studied families, such as the cyclotomics. 

In \cite{Blanco:2021:VMTR, BL:2022:RPE}, the authors show the non-existence of small roots for the maximal real cyclotomic polynomials $\Psi_n(x)$. They prove that for conductors of the form $n = 2^r k$, where $k$ is odd, the polynomials $\Psi_n(x)$ do not have roots $\alpha \in \{\pm 1, \pm 2\}$ modulo any odd prime $q$. Hence, under the given assumptions, the attacks against small order with $\alpha = \pm 1$ or small residue with $\pm 2$ are not effective for the polynomial family $\Psi_n(x)$. On the other hand, for the cyclotomic polynomials $\Phi_n(x)$, $\alpha = 2$ can be a root. However, for this to happen the prime modulus $q$ has to be in the order of $q \approx 10^{18}$, which is too large a modulus in practice.

A computational approach to study the robustness of the maximal real cyclotomic polynomials $\Psi_n(x)$ against the small-root attack was conducted in \cite{ahola2024fast}. The authors ran a massive search for primes $p < 1500$ and $q < 5 \cdot 10^{10}$ to test for small roots $\alpha \in \{\pm 2,\pm 3, \pm 4,\pm 8\}$ of the polynomials $\Psi_{4p}(x)$ and $\Phi_{p}(x)$ modulo $q$. We continue to analyze the robustness of this family by expanding our study to a larger set of attacks and a more dense set of degrees parametrized by $n = 2^r p^s$.

\subsection{Methodology}
In order to analyze whether the maximal real cyclotomic polynomials are subject to these types of attacks, we will follow the subsequent steps:
\begin{enumerate}
    \item The polynomials $\Psi_n(x)$ are constructed using the explicit formulas of Section \ref{sec:minimal_polynomial} and then expanded to be represented in the power basis.
    \item The polynomials to be analyzed are built from \textit{cryptographically relevant} parameters, i.e., from values that are expected to appear or to be used in current cryptographic schemes.
    \item For each polynomial to be evaluated, we extract all the information necessary to carry out the attacks. This includes:
    \begin{enumerate}
        \item Every small-order root ($\text{ord}(\alpha) < 5$) of $\Psi_n(x)$ seen over $\mathbb{F}_q$. 
        \item The smallest-order root of each polynomial in $\mathbb{F}_q$ (if any).
        \item Every $k$-ideal factor of the polynomial ($k < 5$ and $\text{ord}(-a) < 5$).
        \item The smallest-order $k$-ideal factor of each polynomial (if any).
    \end{enumerate}
    \item The ratio of vulnerable instances under the attacks is compared with the same computations for cyclotomic polynomials. In other words, we compare the security of the maximal real polynomials to the more established cyclotomic polynomials.
\end{enumerate}
First, we analyze the polynomials $\Psi_n(x)$ in cases that are similar to the selected structured lattice standards. We build the maximal real cyclotomic polynomials with the same degree and use the same modulus as the first three structured lattice-based standards ML-KEM, ML-DSA, and future FN-DSA. Because it is known that these schemes are not subject to root-based algebraic attacks, this comparison provides the first insight to the safety of the polynomials $\Psi_n(x)$.

Second, we consider a large sample of polynomials and moduli, and the attacks are applied to all possible combinations of the two. We evaluate the likelihood of each of these instances to be vulnerable to these attacks. Taking into account all the samples and all the attacks, this analysis yields a \emph{vulnerability ratio}: how many of the total instances were vulnerable to these attacks. 

Since we are restricted to a statistical approach, the sampled polynomials are constructed under a number of constraints which ensure certain cryptographic robustness, that is, the selection of all instances must have a large enough degree and a large enough prime modulus. The total number of test instances $(\Psi_n(x), q)$ is more than 3500, which is large enough to predict how vulnerable on average the new polynomials $\Psi_n(x)$ are to root-based attacks. For comparison, the associated cyclotomic polynomials of the same degree and modulus are tested against the same attacks, and we provide the vulnerability ratios of the attacks for both families. For the interested reader, all the required material along with computational programs and functions can be found in a Python notebook on Github. \footnote{\url{https://github.com/RodriM11/PLWE-Maximal-Real-Computations}}

\subsection{Analysis based on ML-KEM, ML-DSA and (future) FN-DSA settings}
ML-KEM, ML-DSA and (future) FN-DSA represent the first three standards that are based on a structured lattice assumption. While each of them have their own characteristics and defining details, they all have the additional algebraic structure which is relevant to the security analysis of the scheme.

As a starting point, to analyze the robustness of the polynomials $\Psi_n(x)$ against root-based attacks on PLWE, we use the exact same parameters as in the three standardized schemes:
\begin{itemize}
    \item ML-KEM: degree = 256, $q = 3329$
    \item ML-DSA: degree = 256, $q = 8380417$
    \item (future) FN-DSA (Falcon): degree $\in \{512, 1024\}$, $q = 12289$.
\end{itemize}

To simulate the same parameters for the maximal real cyclotomic polynomials, we need to define $(p, s, r)$ such that the degree of the associated polynomial matches the degree in each setting. In Section \ref{sec:minimal_polynomial}, it was established that $$\deg(\Psi_{2^r p^s}) = (p-1) \cdot p^{s-1} \cdot 2^{r-2}$$ for $r \geq 1, s \geq 1$ and $p \geq 3$ prime. Therefore, to simulate $m \in \{256, 512, 1024\}$, we use the cases 
$$(p, s, r) \in \{(5, 1, 8), (5, 1, 9), (5, 1, 10)\},$$
along with the matching modulus $q$ for each simulated scheme. 

After running the computations, we find that for each of the ML-KEM, ML-DSA, and FN-DSA settings, their respective polynomials $\Psi_n(x)$ do not have any roots in $\mathbb{F}_q$. Also, they do not have any $k$-ideal factors, that is factorizations formed by quadratic (quartic for ML-DSA) polynomials with zero intermediate coefficients. Therefore, the polynomials $\Psi_n(x)$ maintain their security against root-based attacks on the algebraic structure.

These results provide the first important hint that the use of the maximal real cyclotomic polynomials under the same conditions as cryptographically relevant cyclotomic instances is secure against a large range of algebraic attacks on PLWE.

\subsection{Statistical analysis}
The objective of the this section is to build upon the results above and study further the security of the maximal real polynomials against root-based PLWE attacks. To do so, we follow a \textit{statistical approach}. We draw a number of sample instances $(\Psi_n(x),q)$, and consider for each the same attacks as above. 

In particular, we build our sample space for the instances such that all samples satisfy the following restrictions. For the conductor $n = 2^r p^s$, we use all possible combinations of $(p, s, r) \in S$, where $S$ is defined as
$$S = \{(p,s,r) : 5 \leq p \leq 50, \, 2 \leq r \leq 9, \, 1 \leq s \leq 3, \, 256 \leq \deg(\Psi_{2^r p^s}) \leq 512\}.$$
In other words, we restrict to parameter space $p \in [5, 50]$ prime, $r \in [2, 9]$, $s \in [1, 3]$, such that the degree $m = \deg(\Psi_{2^r p^s})$ lands in $[256, 512]$. This yields $|S| = 24$ combinations. Furthermore, we sample 150 primes $q$ randomly in the range $2048 \leq q \leq 4192$ to use as the modulus. Finally, we use six different Gaussian distributions of standard deviation $\sigma \in \{2,3,4,5,6,7\}$. In total, this yields 3600 combinations $(\Psi_n(x), q)$ and 21600 PLWE instances $(\Psi_n(x), q, \sigma)$.

The most relevant results are as follows:
\begin{itemize}
    \item Out of the 3600 combinations for $(\Psi_n(x), q)$, only two pairs have small-order roots and only one has a $k$-ideal factor where $k < 5$.
    \item The vulnerability ratios of the attacks on the smallest-order root and each of the smallest-order $k$-ideal factors are given in Table \ref{tab:root_based_attacks_max_real}.
\end{itemize}
\begin{table}[H]
\caption{Vulnerability ratio of root-based attacks for $\Psi_{2^r p^s}(x)$ for $(p, s, r) \in S$, $q \in [2048, 4096]$ and $\sigma \in [2, 7]$, based on roots and $k$-ideal factors of $\Psi_{2^r p^s}(x)$.}
\label{tab:root_based_attacks_max_real}
\begin{tabular}{l|cccccc}
                        & $\sigma = 2$      & $\sigma = 3$      & $\sigma = 4$      & $\sigma = 5$      & $\sigma = 6$      & $\sigma = 7$ \\\hline
Small set (roots)       & 0                 & 0                 & 0                 & 0                 & 0                 & 0    \\
Small error (roots)     & 0.0003            & 0.0003            & 0.0003            & 0.0003            & 0.0003            & 0.0003 \\  
Small set ($k$-ideal)   & 0                 & 0                 & 0                 & 0                 & 0                 & 0    \\
Small error ($k$-ideal) & 0.0008 & 0.0008 & 0.0008 & 0.0008 & 0.0008 & 0.0008
\end{tabular}
\end{table}
The results in Table \ref{tab:root_based_attacks_max_real} provide computational evidence that these maximal real cyclotomic instances are prone to a very small vulnerability ratios for the attacks. For the attack against small error values, for each $\sigma$, we found only one vulnerable instance giving the displayed vulnerability ratio $1 / 3600 \approx 0.0003$. For the corresponding $k$-ideal attacks, only nine instances are found to be vulnerable, with a vulnerability ratio $9 / (3 \cdot 3600) \approx 0.0008$ for all $\sigma$'s. Here the extra factor of $3$ in the denominator comes from the consideration of three different factors $x^k - a$ with $k = 2,3,4$. These computational results add to the conjectured robustness of the polynomial family and gives a notion of the real utility of the polynomial family $\Psi_n(x)$ in practical PLWE instances.

To provide a measure of how small these values are in comparison, we run the same analysis for the cyclotomic counterparts of the same degree and modulus as above for $\Psi_n(x)$ parametrized by the set $S$. For cyclotomic polynomials, we did not find any instances in our sample space that were be susceptible to the any of the four algebraic attacks. Therefore, the cyclotomic polynomials seem to be more resistant than the maximal real cyclotomic polynomials. However, given that only 10 out of the 3600 maximal real cyclotomic instances were vulnerable, from the practical point of view, the maximal real cyclotomic polynomials are on average as secure as the cyclotomic polynomials in the tested range of parameters.

For the $k=1$ ideal attack, which equals the original root-based approach, the only vulnerable instance was $(p, s, r, q) = (19, 2, 2, 2887)$. It had a root $\alpha = 698$ of order $3$. For the $k$-ideal attack, a few additional vulnerable instances appear. Even though the increase is negligible when compared to the total number of instances evaluated, the increase is consistent with the theoretical analysis that higher degree finite extensions are more likely to be affected by these attacks.

The non-zero vulnerability ratios of the maximal real cyclotomic polynomials are due to the fact that particular instances have smaller order roots. For the cyclotomic polynomials, this does not happen since the elements under consideration in every type of attack (namely, the $k$-ideal factors of the form $x^k - a$, with $k \in \{1,.., 4\}$) are $\frac{n}{\text{gcd}(n, k)}$-th primitive roots of unity. Note that this is due to the fact that, if a cyclotomic polynomial has an irreducible factor of the form $x^k - a$, then this factor encloses $k$ of the $n$-th primitive roots of unity, and every power $\alpha^k$ of a $n$-th primitive root of unity is an $a$-th primitive root of unity, with $a = \frac{n}{\text{gcd}(n, k)}$.

On the other hand, the maximal real cyclotomic polynomials do not enjoy such simple relations for their roots and $k$-ideal factors, and therefore it is possible, as the numerical computations show, that certain instances of $\Psi_n(x)$ and $q$ do yield roots of small order.

In summary, these results show that for the range of samples that we tested, the maximal real cyclotomic polynomials have only a slightly higher vulnerability ratio than their cyclotomic counterparts. In both experiments, the vulnerability ratios are extremely small, which is in line with the theoretical and conjectured robustness of the PLWE problem for these instances. Overall, in both polynomial families the proportion of polynomial--prime pairs that are invulnerable to the algebraic attacks is very high. Thus, for applications we can always discard vulnerable cases, since the small vulnerability ratio implies that it is easy to find instances such that the algebraic attacks do not work. In other words, in practice one can always work with maximal real instances which are not vulnerable against any type of root-based attack.

\section*{Acknowledgements}

W. Bolaños was supported in part by the Research Council of Finland (grant \#351271, PI C. Hollanti). R. M. Sánchez-Ledesma was supported in part by the PQReact project, which has received funding from the Horizon Europe research and innovation program of the European Union under grant agreement \#101119547. Lastly, A. Haavikko was supported by the Jenny and Antti Wihuri Foundation (grant \#00240063). We would like to thank Camilla Hollanti and Iván Blanco-Chacón for their valuable comments on the drafts of this paper.

\bibliographystyle{plain}
\bibliography{references}

\end{document}